\pdfoutput=1
\RequirePackage{ifpdf}
\ifpdf 
\documentclass[pdftex]{sigma}
\else
\documentclass{sigma}
\fi

\numberwithin{equation}{section}

\newtheorem{Theorem}{Theorem}[section]
\newtheorem{Proposition}[Theorem]{Proposition}

\begin{document}

\allowdisplaybreaks

\newcommand{\arXivNumber}{1712.01549}

\renewcommand{\PaperNumber}{017}

\FirstPageHeading

\ShortArticleName{Evolutionary Hirota Type (2+1)-Dimensional Bi-Hamiltonian Equations}

\ArticleName{Evolutionary Hirota Type (2+1)-Dimensional\\ Equations: Lax Pairs, Recursion Operators\\ and Bi-Hamiltonian Structures}

\Author{Mikhail B. SHEFTEL~$^\dag$ and Devrim YAZICI~$^\ddag$}

\AuthorNameForHeading{M.B.~Sheftel and D.~Yaz{\i}c{\i}}

\Address{$^\dag$~Department of Physics, Bo\u{g}azi\c{c}i University, Bebek, 34342 Istanbul, Turkey}
\EmailD{\href{mailto:mikhail.sheftel@boun.edu.tr}{mikhail.sheftel@boun.edu.tr}}

\Address{$^\ddag$~Department of Physics, Y{\i}ld{\i}z Technical University, Esenler, 34220 Istanbul, Turkey}
\EmailD{\href{mailto:yazici@yildiz.edu.tr}{yazici@yildiz.edu.tr}}

\ArticleDates{Received December 06, 2017, in f\/inal form March 02, 2018; Published online March 07, 2018}

\Abstract{We show that evolutionary Hirota type Euler--Lagrange equations in $(2+1)$ dimensions have a symplectic Monge--Amp\`ere form.
We consider integrable equations of this type in the sense that they admit inf\/initely many hydrodynamic reductions and determine Lax pairs for them. For two seven-parameter families of integrable equations converted to two-component form we have constructed Lagrangians, recursion operators and bi-Hamiltonian representations. We have also presented a six-parameter family of tri-Hamiltonian systems.}


\Classification{35Q75; 37K05; 37K10}

\section{Introduction}

We study recursion operators, Lax pairs and bi-Hamiltonian representations for $(2+1)$-dimen\-sional equations of the evolutionary Hirota type
\begin{gather}
 u_{tt} = f(u_{t1},u_{t2},u_{11},u_{12},u_{22}). \label{Hirota}
\end{gather}
Here $u=u(t,z_1,z_2)$ and the subscripts denote partial derivatives of $u$, namely, $u_{ij} = \partial^2 u/\partial z_i\partial z_j$, $u_{ti} = \partial^2 u/\partial t \partial z_i$. Equations of this type arise in a wide range of applications including non-linear physics, general relativity, dif\/ferential geometry and integrable systems. Some examples are the Khokhlov--Zabolotskaya (dKP) equation in non-linear acoustics and the theory of Einstein--Weyl structures and the Boyer--Finley equation in the theory of self-dual gravity.

A lot of work has been done by E.~Ferapontov et~al.\ for studying integrability of equation~\eqref{Hirota} which is understood as the existence of inf\/initely many hydrodynamic reductions (see~\cite{ferdub,ferkhus3,ferkhus2,ferkhus1,fer4} and references therein). Ferapontov et~al.\ in~\cite{ferkhus3} derived integrability condition which is equivalent to the property of equation~\eqref{Hirota} to be linearizable by contact transformations. This does not mean that there is a straightforward way to obtain bi-Hamiltonian structures of the nonlinear equations~\eqref{Hirota} by a direct transfer of such structures from the linear equation, because an arbitrary transformation in the space of second partial derivatives of the unknown will not preserve the bi-Hamiltonian structure of the equation.

Our goal here is to study bi-Hamiltonian structures of the integrable equations of the \linebreak form~\eqref{Hirota} together with Lax pairs and recursion operators. We utilize the method which we used earlier~\cite{nns,sy,sym} for constructing a degenerate Lagrangian for two-component evolutionary form of the equation and using Dirac's theory of constraints~\cite{dirac} in order to obtain Hamiltonian form of the system.

Our approach here starts with description of all equations \eqref{Hirota} which have the Euler--Lagrange form~\cite{olv}. We do not consider here the equations of the form~\eqref{Hirota} which become Lagrangian only after multiplication by an integrating factor of variational calculus, postponing the appropriate generalization for a further publication. We f\/ind that the Lagrangian evolutionary Hirota-type equations have a symplectic Monge--Amp\`ere form and we determine their Lagrangians. Then we study recursion relations for symmetries and Lax pairs for the Lagrangian equations. Here our starting point is to convert the symmetry condition into a ``skew-factorized'' form. This approach extends the method of A.~Sergyeyev~\cite{Artur} for constructing recursion ope\-ra\-tors. According to~\cite{Artur} one constructs the recursion operator from a certain Lax pair which, in turn, is typically built from the original Lax pair for the equation under study. On the other hand, below we construct such a special Lax pair using the skew-factorized form of the linearized equation (symmetry condition) rather than a previously known Lax pair, and then apply the construction from~\cite{Artur} to obtain the recursion operator.

The next step is to transform Lagrangian equations converted into a two-component form to a Hamiltonian system. Finally, composing a recursion operator with a Hamiltonian operator we obtain the second Hamiltonian operator and also f\/ind the corresponding Hamiltonian density.
Thus, we end up with a bi-Hamiltonian representation of an integrable equation~\eqref{Hirota} in a~two-component form. In this way, we obtain two seven-parameter families of bi-Hamiltonian systems and a~six-parameter family of tri-Hamiltonian systems.

The paper is organized as follows. In Section~\ref{Lagrange}, we show that all equations~\eqref{Hirota} of the Euler--Lagrange form have the symplectic Monge--Amp\`ere form and we derive a Lagrangian for such equations. In Section~\ref{symmetry}, we analyze the symmetry condition for the Lagrangian equations. In Section~\ref{recursrel}, using an integrability condition, we convert the symmetry condition into a ``skew-factorized'' form and immediately extract Lax pair and recursion relations for symmetries from the symmetry condition in this form. In Section~\ref{two-comp}, we convert our equation in a two-component form and derive a degenerate Lagrangian for this system. In Section~\ref{Hamilton}, we transform the Lagrangian system into Hamiltonian system using the Dirac's theory of constraints. We obtain the Hamiltonian operator~$J_0$ and corresponding Hamiltonian density~$H_1$. In Section~\ref{recursoper}, we derive a recursion operator~$R$ in a $2\times 2$ matrix form using recursion relations for the two-component form of the equation. In Section~\ref{bi-Ham}, by composing the recursion operator with the Hamiltonian operator $J_0$ we obtain the second Hamiltonian operator $J_1=RJ_0$ and the corresponding Hamiltonian density~$H_0$ with one additional ``Hamiltonian'' constraint for coef\/f\/icients. Thus, we obtain a seven-parameter family of bi-Hamiltonian systems because of the two constraints on nine coef\/f\/icients: integrability condition and Hamiltonian condition.

We consider also an alternative skew-factorized representation of the symmetry condition which implies dif\/ferent Lax pair, another recursion operator and a dif\/ferent seven-parameter family of bi-Hamiltonian systems under a dif\/ferent additional constraint on the coef\/f\/icients. If, in addition, we require that both additional constraints coincide and are compatible with the integrability condition, we obtain a six-parameter family of tri-Hamiltonian systems.

In Sections~\ref{recursrel}, \ref{recursoper}, and \ref{bi-Ham}, we treat separately each of the two generic cases of equation~\eqref{1comp} when none of the coef\/f\/icients $c_1$, $c_2$, $c_3$ vanishes, particular cases when either $c_1=0$ or $c_2=0$ which are obtained as a specialization of one of the generic cases, and the special case $c_3=0$ which cannot be obtained from the generic cases but should be treated independently.

\section{Lagrangian equations of evolutionary Hirota type}\label{Lagrange}

We start with equation \eqref{Hirota} in the form
\begin{gather}
 F \equiv - u_{tt} + f(u_{t1},u_{t2},u_{11},u_{12},u_{22}) = 0. \label{F}
\end{gather}
The Fr\'echet derivative operator (linearization) of equation \eqref{F} reads
\begin{gather*}
 D_F = - D_t^2 + f_{u_{t1}}D_tD_1 + f_{u_{t2}}D_tD_2 + f_{u_{11}}D_1^2 + f_{u_{12}}D_1D_2 + f_{u_{22}}D_2^2, 
\end{gather*}
where $D_i\equiv D_{z_i}, D_t$ denote operators of total derivatives. The adjoint Fr\'echet derivative operator has the form
\begin{gather*}
 D^*_F = - D_t^2 + D_tD_1f_{u_{t1}} + D_tD_2f_{u_{t2}} + D_1^2f_{u_{11}} + D_1D_2f_{u_{12}} + D_2^2f_{u_{22}}. 
\end{gather*}
According to Helmholtz conditions \cite{olv}, equation \eqref{F} is an Euler-Lagrange equation for a~variational problem if\/f its Fr\'echet derivative is self-adjoint, $D^*_F =D_F$, or explicitly, equating to zero coef\/f\/icients of $D_t$, $D_1$, $D_2$ and the term without operators of total derivatives, we obtain four equations on~$f$
\begin{gather*}
 D_1[f_{u_{t1}}] + D_2[f_{u_{t2}}] = 0,\\ 
 D_t[f_{u_{t1}}] + 2D_1[f_{u_{11}}] + D_2[f_{u_{12}}] = 0,\\ 
 D_t[f_{u_{t2}}] + 2D_2[f_{u_{22}}] + D_1[f_{u_{12}}] = 0,\\ 
D_tD_1[f_{u_{t1}}] + D_tD_2[f_{u_{t2}}] + D_1^2[f_{u_{11}}] + D_1D_2[f_{u_{12}}] + D_2^2[f_{u_{22}}] = 0. 
\end{gather*}
The general solution of these equations for $f$ implies the Lagrangian evolutionary Hirota equation~\eqref{Hirota} to have symplectic Monge--Amp\`ere form
\begin{gather}
 u_{tt} = c_1(u_{1t}u_{12}-u_{2t}u_{11}) + c_2(u_{1t}u_{22}-u_{2t}u_{12}) + c_3\big(u_{11}u_{22}-u_{12}^2\big) + c_4u_{1t} + c_5u_{2t}\nonumber\\
\hphantom{u_{tt} =}{} + c_6u_{11} + c_7u_{12} + c_8u_{22} + c_9. \label{1comp}
\end{gather}
A Lagrangian for the equation \eqref{1comp} is readily obtained by applying the homotopy formula \cite{olv} for $F = - u_{tt} + f$
\begin{gather*} L[u] = \int_0^1 u\cdot F[\lambda u]{\rm d}\lambda\end{gather*}
with the result
\begin{gather}
L = -\frac{1}{2}uu_{tt} + \frac{u}{3}\bigl[c_1(u_{1t}u_{12}-u_{2t}u_{11}) + c_2(u_{1t}u_{22}-u_{2t}u_{12}) + c_3\big(u_{11}u_{22}-u_{12}^2\big)\bigr] \nonumber\\
\hphantom{L=}{} + \frac{u}{2}(c_4u_{1t} + c_5u_{2t} + c_6u_{11} + c_7u_{12} + c_8u_{22}) + c_9u.\label{L}
\end{gather}

\section{Symmetry condition}\label{symmetry}

In the following it will be useful to introduce operator of derivative in the direction of the vector $\vec{c}=(c_1,c_2)$: $\nabla_c = \vec{c}\cdot\nabla = c_1D_1 + c_2D_2$, so that equation~\eqref{1comp} may be written as
\begin{gather}
 u_{tt} = u_{1t}\nabla_c(u_2) - u_{2t}\nabla_c(u_1) + c_3\big(u_{11}u_{22}-u_{12}^2\big)\nonumber\\
 \hphantom{u_{tt} =}{} + c_4u_{1t} + c_5u_{2t} + c_6u_{11} + c_7u_{12} + c_8u_{22} + c_9.\label{eqnabla}
\end{gather}

Symmetry condition is the dif\/ferential compatibility condition of \eqref{1comp} and the Lie equation $u_\tau = \varphi$, where
$\varphi$ is the symmetry characteristic and $\tau$ is the group parameter. It has the form of Fr\'echet derivative (linearization) of equation~\eqref{eqnabla}
\begin{gather}
 \varphi_{tt} = \nabla_c(u_2)\varphi_{1t} + u_{1t}\nabla_c(\varphi_2) - \nabla_c(u_1)\varphi_{2t} - u_{2t}\nabla_c(\varphi_1)
\nonumber\\
\hphantom{\varphi_{tt} =}{} + c_3(u_{22}\varphi_{11} + u_{11}\varphi_{22} - 2u_{12}\varphi_{12}) + c_4\varphi_{1t} + c_5\varphi_{2t} + c_6\varphi_{11} + c_7\varphi_{12} + c_8\varphi_{22}.\label{symcond}
\end{gather}
It is convenient to introduce the following dif\/ferential operators
\begin{gather*}
L_{12(1)} = u_{12}D_1 - u_{11}D_2,\qquad L_{12(2)} = u_{22}D_1 - u_{12}D_2,\\
 L_{12(t)} = u_{2t}D_1 - u_{1t}D_2 = v_2D_1 - v_1D_2,\qquad v = u_t,
\end{gather*}
so that the symmetry condition \eqref{symcond} becomes
\begin{gather} \big\{c_1(D_tL_{12(1)} - D_1L_{12(t)}) + c_2(D_tL_{12(2)} - D_2L_{12(t)}) + c_3(D_1L_{12(2)} - D_2L_{12(1)}) - D_t^2\nonumber\\
\qquad{} + c_4D_1D_t + c_5D_2D_t + c_6D_1^2 + c_7D_1D_2 + c_8D_2^2\big\}\varphi = 0.\label{DE}
\end{gather}

\section{Recursion relations and Lax pairs}\label{recursrel}

E.~Ferapontov et al.\ in \cite{ferkhus3} have derived an integrability condition for the symplectic Monge--Amp\`ere equation of the following general form (equation~(21) in~\cite{ferkhus3})
\begin{gather}
 \epsilon \det\left[\begin{matrix}
 u_{11} & u_{12} & u_{13}\\
 u_{12} & u_{22} & u_{23}\\
 u_{13} & u_{23} & u_{33}
 \end{matrix}\right] + h_1\big(u_{22}u_{33} - u_{23}^2\big) + h_2\big(u_{11}u_{33} - u_{13}^2\big) + h_3\big(u_{11}u_{22} - u_{12}^2\big)\nonumber \\
\qquad{} + g_1(u_{11}u_{23} - u_{12}u_{13})+ g_2(u_{22}u_{13} - u_{12}u_{23})
 + g_3(u_{33}u_{12} - u_{13}u_{23})\nonumber\\
 \qquad{} + s_1u_{11} + s_2u_{22}+ s_3u_{33} + \tau_1u_{23} + \tau_2u_{13} + \tau_3u_{12} + \nu = 0 \label{21}
\end{gather}
with constant coef\/f\/icients. Here integrability means that the equation \eqref{21} admits inf\/initely many hydrodynamic reductions~\cite{ferkhus1}. For our evolutionary equation \eqref{1comp} the coef\/f\/icients in \eqref{21} are
\begin{gather}
 h_1= h_2 = h_3 = g_3 = 0,\qquad g_1 = - c_1,\qquad g_2 = c_2,\qquad s_1 = c_6,\qquad s_2 = c_8,\nonumber\\
s_3 = - 1, \qquad\tau_1 = c_5,\qquad \tau_2 = c_4,\qquad \tau_3 = c_7,\qquad \nu = c_9. \label{coeff}
\end{gather}
The integrability condition given in \cite[formula (22)]{ferkhus3} for the equation \eqref{21} has the form
\begin{gather}
 h_1^2s_1^2 + h_2^2s_2^2 + h_3^2s_3^2 + g_1^2s_2s_3 + g_2^2s_1s_3 + g_3^2s_1s_2- 2(h_1h_2s_1s_2 + h_1h_3s_1s_3 + h_2h_3s_2s_3)\nonumber \\
 \qquad {} + 4\epsilon s_1s_2s_3 + 4\nu h_1h_2h_3
 + \epsilon\tau_1\tau_2\tau_3 - \nu g_1g_2g_3 - \epsilon^2\nu^2 - \nu\big(g_1^2h_1 + g_2^2h_2 + g_3^2h_3\big)\nonumber\\
\qquad{} - (g_1\tau_1 + g_2\tau_2 + g_3\tau_3 + 2\epsilon\nu)(h_1s_1 + h_2s_2 + h_3s_3 - \epsilon\nu)\nonumber\\
\qquad{} + 2(g_1h_1s_1\tau_1 + g_2h_2s_2\tau_2 + g_3h_3s_3\tau_3) + \tau_1^2h_2h_3 + \tau_2^2h_1h_3 + \tau_3^2h_1h_2\nonumber\\
\qquad{} - \epsilon(\tau_1^2s_1 + \tau_2^2s_2 + \tau_3^2s_3) + s_1\tau_1g_2g_3 + s_2\tau_2g_1g_3 + s_3\tau_3g_1g_2\nonumber\\
\qquad{} - (g_1h_1\tau_2\tau_3 + g_2h_2\tau_1\tau_3 + g_3h_3\tau_1\tau_2) = 0. \label{integr}
\end{gather}
For the equation \eqref{1comp} due to identif\/ications \eqref{coeff}, the \textit{integrability condition} \eqref{integr} becomes
\begin{gather}
 c_2(c_1c_7 - c_2c_6 + c_3c_4) = c_1(c_1c_8 + c_3c_5) - c_3^2.\label{intcon}
\end{gather}
We show explicitly that any equation of the form \eqref{1comp} satisfying \eqref{intcon} is also integrable in the traditional sense by constructing Lax pair for such an equation.

\subsection{Generic case}

In the integrability condition \eqref{intcon} we assume
\begin{gather*}
 c_1\cdot c_2\cdot c_3\ne 0.
\end{gather*}
The following procedure extends A.~Sergyeyev's method for constructing recursion opera\-tors~\cite{Artur}. Namely, unlike~\cite{Artur}, we start with the skew-factorized form of the symmetry condition and extract from there a Lax pair for symmetries instead of building it from a previously known Lax pair. After that we construct a recursion operator from this newly found Lax pair using Proposition~1 from~\cite{Artur}.

The linear operator of the symmetry condition \eqref{DE} for integrable equations of the form \eqref{1comp} can be presented in the ``skew-factorized'' form
\begin{gather}
 (A_1B_2 - A_2B_1)\varphi = 0. \label{factorDE}
\end{gather}
If we introduce two-dimensional vector operators $\vec{R}=(A_1,A_2)$ and $\vec{S}=(B_1,B_2)$, then the skew-factorized form \eqref{factorDE} becomes
the cross (vector) product $(\vec{R}\times\vec{S})\varphi = 0$. Here dif\/ferential operators $A_i$ and $B_i$ are def\/ined as
\begin{gather}
 A_1 = c_1D_t - c_3D_2,\qquad A_2 = -(c_2D_t + c_3D_1),\nonumber\\
 B_1 = c_1\big(c_3L_{12(2)} - c_1L_{12(t)}\big) + c_1c_6D_1 + (c_1c_7 - c_2c_6 + c_3c_4)D_2, \nonumber \\
 B_2 = c_1\big(c_3L_{12(1)} + c_2L_{12(t)}\big) + (c_3c_4 - c_2c_6)D_1 - c_1c_8D_2 - c_3D_t. \label{AiBi}
\end{gather}
These operators satisfy the commutator relations
\begin{gather}
 [A_1,A_2] = 0,\qquad [A_1,B_2] - [A_2,B_1] = 0,\qquad [B_1,B_2] = 0, \label{commut}
\end{gather}
where the last equation is satisf\/ied on solutions of the equation~\eqref{1comp}.

It immediately follows that the following two operators also commute on solutions
\begin{gather}
 X_1 = \lambda A_1 + B_1,\qquad X_2 = \lambda A_2 + B_2, \qquad [X_1,X_2] = 0, \label{Lax}
\end{gather}
and therefore constitute Lax representation for equation~\eqref{1comp} with $\lambda$ being a spectral parameter.

Symmetry condition in the form \eqref{factorDE} not only provides the Lax pair for equation~\eqref{1comp} but also leads directly to recursion relations for symmetries
\begin{gather}
 A_1\tilde{\varphi} = B_1\varphi,\qquad A_2\tilde{\varphi} = B_2\varphi, \label{recurs}
\end{gather}
where $\tilde{\varphi}$ is a potential for $\varphi$. This follows from a special case of Proposition 1 of \cite{Artur} which we recapitulate here together with its proof for the readers' convenience. Indeed, equations \eqref{recurs} together with~\eqref{commut} imply $(A_1B_2 - A_2B_1)\varphi = [A_1,A_2]\tilde{\varphi}=0$, so $\varphi$ is a symmetry characteristic. Moreover, due to~\eqref{recurs}
\begin{gather*}(A_1B_2 - A_2B_1)\tilde{\varphi} = \bigl([A_1,B_2] - [A_2,B_1] + B_2A_1 - B_1A_2\bigr)\tilde{\varphi} = [B_2,B_1]\varphi = 0, \end{gather*}
which shows that $\tilde{\varphi}$ satisf\/ies the symmetry condition \eqref{factorDE} and hence is also a symmetry. Thus, $\varphi$ is a symmetry whenever so is $\tilde{\varphi}$ and vice versa.
The equations \eqref{recurs} def\/ine an auto-B\"acklund transformation between the symmetry conditions written for $\varphi$ and $\tilde{\varphi}$.
Hence, the auto-B\"acklund transformation for the symmetry condition is nothing else than a recursion operator.
Finally, we must remark that this approach to recursion operators originated from the much older work published in 90th \cite{gut,Marvan, papa}.

The skew-factorized form of the symmetry condition is by no means unique. We can derive another version by using the discrete symmetry transformation
\begin{gather}
 c_1\leftrightarrow -c_2,\qquad c_4\leftrightarrow c_5,\qquad c_6\leftrightarrow c_8,\qquad D_1\leftrightarrow - D_2,\qquad D_t\leftrightarrow -D_t,\qquad v\leftrightarrow -v,\nonumber\\
L_{12(1)} \leftrightarrow L_{12(2)}, \label{discrete}
\end{gather}
while the operator $L_{12(t)}$ is not changed. Applying \eqref{discrete} to operators~\eqref{AiBi} we obtain a new set of operators
\begin{gather}
 A_1 = c_2D_t + c_3D_1,\qquad A_2 = c_3D_2 - c_1D_t,\nonumber\\
 B_1 = -c_2\big(c_3L_{12(1)} + c_2L_{12(t)}\big) + (c_2c_7 - c_1c_8 - c_3c_5)D_1 + c_2c_8D_2, \nonumber \\
 B_2 = c_2\big(c_1L_{12(t)} - c_3L_{12(2)}\big) - c_2c_6D_1 - (c_1c_8 + c_3c_5)D_2 + c_3D_t, \label{AiBi2}
\end{gather}
which also satisfy the skew-factorized form~\eqref{factorDE} of the symmetry condition \eqref{DE} and the same commutator relations~\eqref{commut}. Using these operators in \eqref{Lax} and~\eqref{recurs} we obtain the second Lax pair and another set of recursion relations for symmetries, respectively.

It may be interesting to note algebraic relations between the two recursion operators, namely, determined by the set \eqref{AiBi}, marked below with the superscript~$^{(1)}$, and by the set~\eqref{AiBi2}, marked with the superscript~$^{(2)}$
\begin{gather*}
 A_1^{(2)} = - A_2^{(1)},\qquad A_2^{(2)} = - A_1^{(1)},\\ c_2B_1^{(1)} + c_1B_2^{(2)} = c_3A_1^{(1)},\qquad c_2B_2^{(1)} + c_1B_1^{(2)} = c_3A_2^{(1)}. 
\end{gather*}

\subsection{Particular cases}

We now consider particular cases $c_1=0$, $c_2\ne 0$ and $c_2=0$, $c_1\ne 0$.

In the f\/irst case integrability condition \eqref{intcon} reads
\begin{gather}
 c_1 = 0,\qquad c_2 \ne 0\quad \Longrightarrow \quad c_2(c_2c_6 - c_3c_4) = c_3^2. \label{spec2}
\end{gather}
If we set $c_1=0$ in our f\/irst set of operators \eqref{AiBi} from the generic case, these operators become linearly dependent and the skew-factorized representation \eqref{factorDE} does not reproduce the symmetry condition \eqref{symcond}. Therefore, we have to put $c_1 = 0$ in the second set of operators \eqref{AiBi2} with the result
\begin{gather}
 A_1 = c_2D_t + c_3D_1,\qquad A_2 = c_3D_2,\nonumber\\
B_1 = -c_2\big(c_3L_{12(1)} + c_2L_{12(t)}\big) + (c_2c_7 - c_3c_5)D_1 + c_2c_8D_2, \nonumber \\
B_2 = c_3D_t - c_2c_3L_{12(2)} - c_2c_6D_1 - c_3c_5D_2. \label{AiBic1=0}
\end{gather}

 Operators \eqref{AiBic1=0} satisfy skew-factorized form \eqref{factorDE} of the symmetry condition \eqref{DE} and the commutator relations \eqref{commut}. Therefore, as is shown above, equations \eqref{recurs} yield the recursion relations for symmetries
\begin{gather*}
(c_2D_t + c_3D_1)\tilde{\varphi} = \big\{{-}c_2\big(c_3L_{12(1)} + c_2L_{12(t)}\big) + (c_2c_7 - c_3c_5)D_1 + c_2c_8 D_2\big\}\varphi,\\
c_3D_2\tilde{\varphi} =\big(c_3D_t - c_2c_3L_{12(2)} - c_2c_6D_1 - c_3c_5D_2\big)\varphi, 
\end{gather*}
and the operators
\begin{gather*}
X_1 = \lambda (c_2D_t + c_3D_1) - c_2(c_3L_{12(1)} + c_2L_{12(t)}) + (c_2c_7 - c_3c_5)D_1 + c_2c_8D_2,\\
X_2 = \lambda c_3D_2 + c_3D_t - c_2c_3L_{12(2)} + c_2c_6D_1 + c_3c_5D_2 
\end{gather*}
commute on solutions and so constitute Lax pair for the equation \eqref{1comp} at $c_1=0$.

In the second case $c_2=0$, integrability condition \eqref{intcon} has the form
\begin{gather}
 c_2 = 0,\qquad c_1 \ne 0\quad \Longrightarrow \quad c_1(c_1c_8 + c_3c_5) = c_3^2. \label{spec2'}
\end{gather}
We set $c_2=0$ in our f\/irst set of operators \eqref{AiBi} in the generic case to obtain
\begin{gather}
A_1 = c_1D_t - c_3D_2,\qquad A_2 = - c_3D_1 ,\nonumber\\
B_1 = c_1\big(c_3L_{12(2)} - c_1L_{12(t)}\big) + c_1c_6D_1 + (c_1c_7 + c_3c_4)D_2, \nonumber \\
B_2 = c_1c_3L_{12(1)} + c_3c_4D_1 - c_1c_8D_2 - c_3D_t. \label{AiBic2=0}
\end{gather}
Operators \eqref{AiBic2=0} also satisfy skew-factorized form \eqref{factorDE} of the symmetry condition \eqref{DE} and the commutator relations~\eqref{commut}. Therefore, as is shown above, using these operators in~\eqref{Lax} and~\eqref{recurs} we obtain the Lax pair and recursion relations for symmetries in this case.

We could not use the second family of operators \eqref{AiBi2} at $c_2=0$ because skew-factorized form \eqref{factorDE} would give identical zero instead of reproducing the symmetry condition.

We note again that the skew-factorized form of the invariance condition \eqref{factorDE} is still not unique. One could obtain such a form at $c_2=0$ with a dif\/ferent choice of the operators~$A_i$ and~$B_i$
\begin{gather}
 A_1 = c_1D_1,\qquad A_2 = c_3D_2 - c_1D_t,\nonumber\\
 B_1 = c_1^2L_{12(1)} + (c_1c_5 - c_3)D_2 - c_1D_t, \nonumber \\
 B_2 = c_1\big\{c_3L_{12(2)} - c_1L_{12(t)} + c_4D_t +c_6D_1 + c_7D_2\big\}. \label{AiBispec2a}
\end{gather}
Recursion relations and the Lax pair are still valid with the new def\/initions \eqref{AiBispec2a}. Since the latter choice leads to more complicated recursion operator and second Hamiltonian operator, we stick to our previous def\/initions~\eqref{AiBic2=0}.

Finally, the case $c_1 = 0$ together with $c_2 = 0$ implies $c_3 = 0$ from the integrability condi\-tion~\eqref{intcon} and hence corresponds to linear equation \eqref{1comp}.

\subsection[Special case $c_3=0$]{Special case $\boldsymbol{c_3=0}$}

In the integrability condition \eqref{intcon} we assume
\begin{gather*}
 c_1\cdot c_2\ne 0,\qquad c_3 = 0\quad \Longrightarrow \quad c_1c_2c_7 = c_1^2c_8 + c_2^2c_6.
\end{gather*}
Then the linear operator of the symmetry condition \eqref{DE} can be presented in the skew-factorized form \eqref{factorDE}
with the following operators $A_i$ and $B_i$
\begin{gather}
 A_1 = c_1D_t,\qquad A_2 = c_1D_1 + c_2D_2 = \nabla_c,\nonumber\\
 B_1 = c_1c_2L_{12(t)} - c_2c_6D_1 - c_1c_8D_2, \nonumber \\
 B_2 = c_2\big\{c_1L_{12(1)} + c_2L_{12(2)} + c_4D_1 + c_5D_2 - D_t\big\}. \label{AiBispec1}
\end{gather}
Operators \eqref{AiBispec1} satisfy the commutator relations \eqref{commut} and hence, as is shown above, the equations~\eqref{recurs} produce the recursion relations for symmetries
\begin{gather}
 c_1D_t\tilde{\varphi} =\{c_1c_2L_{12(t)} - c_2c_6D_1 - c_1c_8D_2\}\varphi, \nonumber\\
\nabla_c(\tilde{\varphi}) = c_2\{c_1L_{12(1)} + c_2L_{12(2)}) + c_4D_1 + c_5D_2 - D_t\}\varphi, \label{recurs1}
\end{gather}
and the operators
\begin{gather*}
X_1 = \lambda c_1D_t + c_1c_2L_{12(t)} - c_2c_6D_1 - c_1c_8D_2,\\
X_2 = \lambda \nabla_c + c_2\big\{c_1L_{12(1)} + c_2L_{12(2)} + c_4D_1 + c_5D_2 - D_t\big\}. 
\end{gather*}
commute on solutions and so constitute Lax representation for the equation \eqref{1comp} at $c_3=0$.

We note that if we considered this case as a particular case at $c_3=0$ of the operators~\eqref{AiBi} or~\eqref{AiBi2} of the generic case, then operators $A_i$, $B_i$ would be linearly dependent and the skew-factorized form~\eqref{factorDE} would not yield the symmetry condition~\eqref{DE}.

\section{Two-component form}\label{two-comp}

Introducing the additional dependent variable $v=u_t$, we convert equation \eqref{eqnabla} into the evolutionary system
\begin{gather}
 u_t = v,\nonumber\\
 v_t = v_1\nabla_c(u_2) - v_2\nabla_c(u_1) + c_3\big(u_{11}u_{22}-u_{12}^2\big)\nonumber\\
\hphantom{v_t =}{} + c_4v_1 + c_5v_2 + c_6u_{11} + c_7u_{12} + c_8u_{22} + c_9.\label{2comp}
\end{gather}
Lie equations become $u_\tau = \varphi$, $v_\tau = \psi$, so that $u_t = v$ implies the f\/irst invariance condition $\varphi_t = \psi$.

The second invariance condition is obtained by dif\/ferentiating the second equation in \eqref{2comp} with respect to the group parameter $\tau$
\begin{gather*}
 \psi_t = \nabla_c(u_2)\psi_1 + v_1\nabla_c(\varphi_2) - \nabla_c(u_1)\psi_2 - v_2\nabla_c(\varphi_1)\\
\hphantom{\psi_t =}{} + c_3(u_{22}\varphi_{11} + u_{11}\varphi_{22} - 2u_{12}\varphi_{12}) + c_4\psi_1 + c_5\psi_2
 + c_6\varphi_{11} + c_7\varphi_{12} + c_8\varphi_{22}.
\end{gather*}

The Lagrangian for system \eqref{2comp} is obtained by a suitable modif\/ication of the Lagran\-gian~\eqref{L} of the one-component equation \eqref{1comp}, skipping some total derivative terms
\begin{gather}
 L = u_tv - \frac{v^2}{2} + \frac{u_t}{3} \{c_1(u_2u_{11} - u_1u_{12}) + c_2(u_2u_{12} - u_1u_{22})\}\nonumber \\
 \hphantom{L =}{} - \frac{u_t}{2}(c_4u_1 + c_5u_2) + \frac{c_3}{3}u\big(u_{11}u_{22} - u_{12}^2\big) + \frac{u}{2}(c_6u_{11} + c_7u_{12} + c_8u_{22}) + c_9u.\label{L2}
\end{gather}

\section{Hamiltonian representation}\label{Hamilton}

To transform from Lagrangian to Hamiltonian description, we def\/ine canonical momenta
\begin{gather}
 \pi_u = \frac{\partial L}{\partial u_t} = v + \frac{1}{3}\bigl\{c_1(u_2u_{11} - u_1u_{12}) + c_2(u_2u_{12} - u_1u_{22})\bigr\} - \frac{1}{2}(c_4u_1 + c_5u_2), \nonumber\\
\pi_v = \frac{\partial L}{\partial v_t} = 0, \label{mom}
\end{gather}
which satisfy canonical Poisson brackets
\begin{gather*} \big[\pi_i(z),u^k(z')\big] = \delta_i^k \delta(z-z'),\end{gather*}
where $u^1=u$, $u^2=v$, $z=(z_1,z_2)$, the only nonzero Poisson bracket being $[\pi_u,u] = \delta(z_1-z_1')\delta(z_2-z_2')$. The Lagrangian~\eqref{L2} is degenerate because the momenta cannot be inverted for the velocities. Therefore, following the theory of Dirac's constraints~\cite{dirac}, we impose~\eqref{mom} as constraints
\begin{gather*}
\Phi_u = \pi_u - v - \frac{1}{3}\bigl[c_1(u_2u_{11} - u_1u_{12}) + c_2(u_2u_{12} - u_1u_{22})\bigr] + \frac{1}{2}(c_4 u_1 + c_5 u_2),\nonumber\\
 \Phi_v = \pi_v, 
\end{gather*}
and calculate Poisson bracket for the constraints
\begin{gather*}
K_{11} = [\Phi_u(z_1,z_2),\Phi_{u'}(z'_1,z'_2)], \qquad K_{12} = [\Phi_u(z_1,z_2),\Phi_{v'}(z'_1,z'_2)], \nonumber\\
K_{21} = [\Phi_v(z_1,z_2),\Phi_{u'}(z'_1,z'_2)],\qquad K_{22} = [\Phi_v(z_1,z_2),\Phi_{v'}(z'_1,z'_2)]. 
\end{gather*}
We obtain the following matrix of Poisson brackets, which for convenience we multiply by the overall factor~$(-1)$
\begin{gather}
 K = \left(
 \begin{matrix}
 - K_{11} & - 1 \\
 1 & 0
 \end{matrix}
 \right),
\label{K}
\end{gather}
where
\begin{gather*}
 K_{11} = c_1(u_{11}D_2 - u_{12}D_1) + c_2(u_{12}D_2 - u_{22}D_1) - c_4D_1 - c_5D_2. 
\end{gather*}
The Hamiltonian operator is the inverse to the symplectic operator $J_0 = K^{-1}$
 \begin{gather}
 J_0 = \left(
 \begin{matrix}
 0 & 1\\
 -1 & - K_{11}
 \end{matrix}
 \right) =
 \left(
 \begin{matrix}
 0 & 1\\
 -1 & c_1L_{12(1)} + c_2L_{12(2)} + c_4D_1 + c_5D_2
 \end{matrix}
 \right).\label{hamilton1}
\end{gather}

Operator $J_0$ is Hamiltonian if and only if its inverse $K$ is symplectic \cite{ff}, which means that the volume integral $\Omega = \iiint_V\omega {\rm d}V$ of $\omega = (1/2){\rm d}u^i\wedge K_{ij}{\rm d}u^j$ should be a symplectic form, i.e., at appropriate boundary conditions ${\rm d}\omega = 0$ modulo total divergence. Another way of formulation is to say that the vertical dif\/ferential of $\omega$ should vanish~\cite{Kras}.
In $\omega$ summations over $i$, $j$ run from~1 to~2 and $u^1 = u$, $u^2 = v$. Using~\eqref{K}, we obtain
\begin{gather}
\omega = \frac{1}{2}\bigl[c_1(u_{12}{\rm d}u\wedge {\rm d}u_1 - u_{11}{\rm d}u\wedge {\rm d}u_2) + c_2(u_{22}{\rm d}u\wedge {\rm d}u_1 - u_{12}{\rm d}u\wedge {\rm d}u_2) \nonumber\\
\hphantom{\omega =}{} - c_4{\rm d}u\wedge {\rm d}u_1 - c_5{\rm d}u\wedge {\rm d}u_2 - 2{\rm d}u\wedge {\rm d}v \bigr]. \label{omega}
\end{gather}
Taking exterior derivative of \eqref{omega} and skipping total divergence terms, we have checked that ${\rm d}\omega=0$ which proves that operator~$K$ is symplectic and hence~$J_0$ def\/ined in~\eqref{hamilton1} is indeed a~Hamiltonian operator.

The f\/irst Hamiltonian form of this system is
\begin{gather*}
\left(
\begin{matrix}
 u_t\\
 v_t
\end{matrix}
\right) = J_0
\left(
\begin{matrix}
 \delta_u H_1\\
 \delta_v H_1
\end{matrix}
\right), 
 \end{gather*}
where we still need to determine the corresponding Hamiltonian density $H_1$. We convert~$L$ from~\eqref{L2} to the form
\begin{gather*}
 L = u_t\pi_u - \frac{v^2}{2} + \frac{c_3}{3}u\big(u_{11}u_{22} - u_{12}^2\big) + \frac{u}{2}(c_6u_{11} + c_7u_{12} + c_8u_{22})+ c_9u, 
\end{gather*}
and apply the formula $H_1 = \pi_u u_t + \pi_v v_t - L$, where $\pi_v=0$, with the f\/inal result
\begin{gather*}
 H_1 = \frac{v^2}{2} - \frac{c_3}{3}u\big(u_{11}u_{22} - u_{12}^2\big) - \frac{1}{2} u(c_6u_{11} + c_7u_{12} + c_8u_{22}) - c_9u.
\end{gather*}

\section[Recursion operators in $2\times 2$ matrix form]{Recursion operators in $\boldsymbol{2\times 2}$ matrix form}\label{recursoper}

\subsection[Generic case: $c_1\cdot c_2\cdot c_3\ne 0$]{Generic case: $\boldsymbol{c_1\cdot c_2\cdot c_3\ne 0}$}

We def\/ine two-component symmetry characteristic $(\varphi,\psi)^{\rm T}$ (where $^{\rm T}$ means transposed matrix) with $\psi=\varphi_t$ and $(\tilde{\varphi},\tilde{\psi})^{\rm T}$ with $\tilde{\psi}=\tilde{\varphi}_t$ for the original and transformed symmetries, respectively. The recursion relations \eqref{recurs} with the use of \eqref{AiBi} for the operators~$A_i$,~$B_i$ take the form
\begin{gather}
 c_1\tilde{\psi} - c_3D_2\tilde{\varphi} = B_1\varphi \equiv \big\{c_1(c_3L_{12(2)} - c_1L_{12(t)}) + c_1c_6 D_1 + (c_3c_4 + c_1c_7 - c_2c_6)D_2\big\}\varphi, \!\!\!\label{phipsi}\\
 -c_2\tilde{\psi} - c_3D_1\tilde{\varphi} = B_2\varphi
\equiv \big\{c_1(c_3L_{12(1)} + c_2L_{12(t)}) + (c_3c_4 - c_2c_6)D_1 - c_1c_8D_2\big\}\varphi - c_3\psi.\nonumber
\end{gather}
Combining these two equations to eliminate f\/irst $\tilde{\psi}$ and then $\tilde{\varphi}$, we obtain
\begin{gather*}
\nabla_c(\tilde{\varphi}) = c_1\big\{\nabla_c(u_1)\varphi_2 - \nabla_c(u_2)\varphi_1\big\} - c_1c_4\varphi_1 - (c_1c_5-c_3)\varphi_2 + c_1\psi,\nonumber \\
\nabla_c(\tilde{\psi}) = c_1\big\{\nabla_c(v_1\varphi_2 - v_2\varphi_1) + c_3(u_{22}\varphi_{11} + u_{11}\varphi_{22} - 2u_{12}\varphi_{12})\nonumber \\
\hphantom{\nabla_c(\tilde{\psi}) =}{} + c_6\varphi_{11} + c_7\varphi_{12} + c_8\varphi_{22}\big\} + c_3\psi_2, 
\end{gather*}
where the subscripts of $\varphi$ and $\psi$ denote partial derivatives. Applying the inverse operator $\nabla_c^{-1}$, which is def\/ined to satisfy the relations $\nabla_c^{-1}\nabla_c = 1$, we obtain the explicit form of recursion relations
\begin{gather}
\tilde{\varphi} = \nabla_c^{-1} \big\langle
 c_1\big\{\nabla_c(u_1)\varphi_2 - \nabla_c(u_2)\varphi_1 - c_4\varphi_1 - c_5\varphi_2\big\} + c_3\varphi_2 + c_1\psi\big\rangle, \nonumber \\
 \tilde{\psi} = c_1(v_1\varphi_2 - v_2\varphi_1) + \nabla_c^{-1} \big\langle c_1\big\{c_3(u_{22}\varphi_{11} + u_{11}\varphi_{22} - 2u_{12}\varphi_{12}) \nonumber \\
\hphantom{\tilde{\psi} =}{} + c_6\varphi_{11} + c_7\varphi_{12} + c_8\varphi_{22}\big\} + c_3\psi_2\big\rangle. \label{tildephipsi}
\end{gather}
Here an important remark is due. The operator $\nabla_c^{-1}$ can make sense merely as a \textit{formal} inverse of~$\nabla_c$. Thus, the relations~\eqref{tildephipsi} are formal as well. The proper interpretation of the quantities like $\nabla_c^{-1}$ and of~\eqref{tildephipsi} requires the language of dif\/ferential coverings, see the original papers~\smash{\cite{gut,Marvan}} and the recent survey~\cite{Kras}.

In a two component form, the recursion relations \eqref{phipsi} read
\begin{gather}
 \left(\begin{matrix}
 \tilde{\varphi}\\
 \tilde{\psi}
 \end{matrix}
 \right) = R\left(
 \begin{matrix}
 \varphi\\
 \psi
 \end{matrix}\right)\label{recurs2}
\end{gather}
with the recursion operator $R$ in the $2\times 2$ matrix form
\begin{gather}
R = \left(
\begin{matrix}
 R_{11} & c_1\nabla_c^{-1}\\
 R_{21} & c_3\nabla_c^{-1}D_2
\end{matrix}\right)\label{R}
\end{gather}
with the matrix elements
\begin{gather}
R_{11} = \nabla_c^{-1} \big\langle
 c_1\big\{\nabla_c(u_1)D_2 - \nabla_c(u_2)D_1 - c_4D_1 - c_5D_2\big\} + c_3D_2\big\rangle, \nonumber\\
R_{21} = c_1 \big\langle v_1D_2 - v_2D_1 + \nabla_c^{-1}\big\{c_3\big(u_{22}D_1^2 + u_{11}D_2^2 - 2u_{12}D_1D_2\big) \nonumber\\
 \hphantom{R_{21} =}{} + c_6D_1^2 + c_7D_1D_2 + c_8D_2^2\big\}\big\rangle. \label{Rij}
\end{gather}

Next, we use the alternative set \eqref{AiBi2} of operators $A_i$, $B_i$ in the recursion relations \eqref{recurs} presented in a two-component form
\begin{gather*}
 c_2\tilde{\psi} + c_3D_1\tilde{\varphi} = B_1\varphi \equiv \big\{{-}c_2(c_3L_{12(1)} + c_2L_{12(t)}) + (c_2c_7 - c_1c_8 - c_3c_5) D_1 + c_2c_8 D_2\big\}\varphi, \\
c_3D_2\tilde{\varphi} - c_1\tilde{\psi} = B_2\varphi
 \equiv \big\{c_2(c_1L_{12(t)} - c_3L_{12(2)}) - c_2c_6D_1 - (c_1c_8 + c_3c_5)D_2\big\}\varphi + c_3\psi.
\end{gather*}
Combining these two equations to eliminate f\/irst $\tilde{\psi}$ and then $\tilde{\varphi}$, we obtain the explicit form of recursion relations
\begin{gather*}
 \tilde{\varphi} = \nabla_c^{-1} \big\{
 {-} c_2(c_1L_{12(1)} + c_2L_{12(2)}) - (c_2c_4+c_3)D_1 - c_2c_5D_2\big\}\varphi + c_2\nabla_c^{-1}\psi, \nonumber \\
 \tilde{\psi} = \left\langle c_3\nabla_c^{-1}D_1\left\{c_1L_{12(1)} + c_2L_{12(2)} + \left(c_4 + \frac{c_3}{c_2}\right)D_1 + c_5D_2\right\} \right.
 \\
 \left.\hphantom{\tilde{\psi} =}{} + \left\{ - (c_3L_{12(1)} + c_2L_{12(t)}) + \frac{1}{c_2} (c_2c_7-c_1c_8-c_3c_5)D_1 + c_8D_2\right\}\right\rangle\varphi - c_3\nabla_c^{-1}D_1\psi,
\end{gather*}
and
we immediately extract the second recursion operator $R'$ in the $2\times 2$ matrix form
\begin{gather}
 R' = \left(\begin{matrix}
 R'_{11} & c_2\nabla_c^{-1} \\
 R'_{21} & - c_3\nabla_c^{-1}D_1
 \end{matrix} \right), \label{R'}
\end{gather}
where
\begin{gather}
 R'_{11} = - c_2\nabla_c^{-1} \left\{
 c_1L_{12(1)} + c_2L_{12(2)} + \left(c_4 + \frac{c_3}{c_2}\right)D_1 + c_5D_2\right\}, \nonumber \\
 R'_{21} = c_3\nabla_c^{-1}D_1\left\{c_1L_{12(1)} + c_2L_{12(2)} + \left(c_4 + \frac{c_3}{c_2}\right)D_1 + c_5D_2\right\} \nonumber \\
\hphantom{R'_{21} =}{} - (c_3L_{12(1)} + c_2L_{12(t)}) + \frac{1}{c_2} (c_2c_7-c_1c_8-c_3c_5)D_1 + c_8D_2. \label{R'ij}
\end{gather}

\subsection{Particular cases}

We consider again particular cases $c_1=0$, $c_2\ne 0$ and $c_2=0$, $c_1\ne 0$.

As we know from Section~\ref{recursrel}, the f\/irst case, $c_1=0$, $c_2\ne 0$ should be considered as a particular case of the second recursion operator $R'$ from~\eqref{R'} and~\eqref{R'ij} with the result
\begin{gather*}
 R' = \left(
 \begin{matrix}
 \displaystyle - c_2D_2^{-1}\!\left(L_{12(2)} + \frac{c_6}{c_3}D_1\right) - c_5, & D_2^{-1} \vspace{1mm}\\
 R'_{21}, & -\displaystyle\frac{c_3}{c_2}D_2^{-1}D_1
 \end{matrix}
 \right), 
\end{gather*}
where
\begin{gather*}
 R'_{21} = c_3 D_2^{-1}D_1\left(L_{12(2)} + \frac{c_6}{c_3}D_1\right) -
 \big(c_3L_{12(1)} + c_2L_{12(t)}\big) + c_7D_1 + c_8D_2 
\end{gather*}
and the integrability condition \eqref{spec2} has been used.

The second case, $c_2=0$, $c_1\ne 0$ should be considered as a particular case of the f\/irst recursion operator $R$ from \eqref{R} and \eqref{Rij} with the result
\begin{gather*}
R = \left(
\begin{matrix}
 R_{11} & D_1^{-1}\vspace{1mm}\\
 R_{21} & \displaystyle \frac{c_3}{c_1}D_1^{-1}D_2
\end{matrix}\right)
\end{gather*}
with the matrix elements
\begin{gather*}
 R_{11} = - D_1^{-1} \big\{c_1L_{12(1)} + (c_1c_5 - c_3)D_2\big\} - c_1c_4,\\
 R_{21} = - c_1 L_{12(t)} + c_6D_1 + c_7D_2 + D_1^{-1}\big\{c_3\big(u_{22}D_1^2 + u_{11}D_2^2 - 2u_{12}D_1D_2\big) + c_8D_2^2\big\}.
\end{gather*}

\subsection[Special case: $c_1\cdot c_2\ne 0$, $c_3 = 0$]{Special case: $\boldsymbol{c_1\cdot c_2\ne 0}$, $\boldsymbol{c_3 = 0}$}

Recursion relations \eqref{recurs1} in a two-component form become
\begin{gather}
 c_1\tilde{\psi} =c_1c_2(v_2\varphi_1 - v_1\varphi_2) - c_2c_6\varphi_1 - c_1c_8\varphi_2, \nonumber\\
 \nabla_c(\tilde{\varphi}) = c_2\big\{\nabla_c(u_2)\varphi_1 - \nabla_c(u_1)\varphi_2 + c_4\varphi_1 + c_5\varphi_2 - \psi\big\}. \label{recurs1a}
\end{gather}
The explicit two-component form of the recursion relations \eqref{recurs1a} is
\begin{gather}
\tilde{\varphi} = \nabla_c^{-1}c_2\big\{\nabla_c(u_2)\varphi_1 - \nabla_c(u_1)\varphi_2 + c_4\varphi_1 + c_5\varphi_2 - \psi\big\},\nonumber\\
\tilde{\psi} = c_2(v_2\varphi_1 - v_1\varphi_2) - \frac{c_2c_6}{c_1}\varphi_1 - c_8\varphi_2. \label{recurs1b}
\end{gather}
In the matrix form \eqref{recurs2}, the recursion operator arising from \eqref{recurs1b} reads
\begin{gather}
 R = \left(
\begin{matrix}
 c_2\nabla_c^{-1}\{\nabla_c(u_2)D_1 - \nabla_c(u_1)D_2 + c_4D_1 + c_5D_2\}, & - c_2\nabla_c^{-1}\vspace{1mm}\\
 c_2(v_2D_1 - v_1D_2) - \displaystyle\frac{c_2c_6}{c_1}D_1 - c_8D_2, & 0
\end{matrix}\right). \label{R1}
\end{gather}

\section{Bi-Hamiltonian systems}\label{bi-Ham}

\subsection{Generic case}

\subsubsection{First family of bi-Hamiltonian systems}

The second Hamiltonian operator $J_1$ is obtained by composing the recursion operator~\eqref{R} with the f\/irst Hamiltonian operator $J_1 = RJ_0$
\begin{gather*}
\left(
\begin{matrix}
 J_1^{11} & J_1^{12} \\
 J_1^{21} & J_1^{22}
\end{matrix}
\right) = \left(
\begin{matrix}
 R_{11} & c_1\nabla_c^{-1} \\
 R_{21} & c_3\nabla_c^{-1}D_2
\end{matrix}
\right) \left(
\begin{matrix}
 0 & 1 \\
 - 1 & \nabla_c(u_2)D_1 - \nabla_c(u_1)D_2 + c_4D_1 + c_5D_2
\end{matrix}
\right),
\end{gather*}
where we have used an alternative equivalent expression for the matrix element $J_0^{22}$, with the f\/inal result
\begin{gather}
J_1 = \left(\begin{matrix}
 -c_1\nabla_c^{-1} & c_3\nabla_c^{-1}D_2 \\
 -c_3\nabla_c^{-1}D_2 & J_1^{22}
\end{matrix}\right),\label{J_1}
\end{gather}
where
\begin{gather*}
 J_1^{22} = c_3L_{12(2)} - c_1L_{12(t)} + \nabla_c^{-1}\big\{c_1\big(c_6D_1^2 + c_7D_1D_2 + c_8D_2^2\big) + c_3D_2(c_4D_1+c_5D_2)\big\}. 
\end{gather*}
Here operator $J_1$ is manifestly skew symmetric. A check of the Jacobi identities and compatibility of the two Hamiltonian structures~$J_0$ and~$J_1$ is straightforward but too lengthy to be presented here. The method of the functional multi-vectors for checking the Jacobi identity and the compatibility of the Hamiltonian operators is developed by P.~Olver in \cite[Chapter~7]{olv} and has been applied recently for checking bi-Hamiltonian structure of the general heavenly equation~\cite{sym} and the f\/irst heavenly equation of Pleba\'nski~\cite{sy} under the well-founded conjecture that this method is applicable for nonlocal Hamiltonian operators as well.

The next problem is to derive the Hamiltonian density $H_0$ corresponding to the second Hamiltonian operator $J_1$ such that implies the bi-Hamiltonian representation of the system~\eqref{2comp}
\begin{gather}
\left(
\begin{matrix}
 u_t\\
 v_t
\end{matrix}
\right) = J_0
\left(
\begin{matrix}
 \delta_u H_1\\
 \delta_v H_1
\end{matrix}
\right) = J_1
\left(
\begin{matrix}
 \delta_u H_0\\
 \delta_v H_0
\end{matrix}
\right) = \left(
\begin{matrix}
 v\\
 v_t
\end{matrix}
\right), \label{be_Ham}
 \end{gather}
where $v_t$ should be replaced by the right-hand side of the second equation in~\eqref{2comp}. Then we could conclude that our system~\eqref{2comp} is also integrable in the sense of Magri~\cite{magri,magri+}.
\begin{Proposition}
Bi-Hamiltonian representation \eqref{be_Ham} of the system \eqref{2comp} is valid under the constraint
\begin{gather}
 c_9 = \frac{1}{c_1^2}\big\{c_6(c_3-c_1c_5) + c_4(c_1c_7 - c_2c_6 + c_3c_4)\big\} \label{c9}
\end{gather}
 with the following Hamiltonian density
\begin{gather}
H_0 = v\left\{\frac{1}{c_1} \nabla_c(u) + \frac{c_4}{c_1}z_2 + \frac{(c_3 - c_1c_5)}{c_1^2}z_1 + s_0\right\} - \frac{c_3}{2c_1^2}u_2\nabla_c(u) + \frac{c_3c_4}{c_1^2} u. \label{H0full}
\end{gather}
\end{Proposition}
\begin{proof} We will need the following simple relation between $J_1^{22}$ and the operator $B_1$ from \eqref{AiBi}, which is involved in the recursion relations~\eqref{recurs} and Lax pair~\eqref{Lax}
\begin{gather}
 J_1^{22} = \frac{1}{c_1}\big(B_1 + \nabla_c^{-1}c_3^2D_2^2\big). \label{J1B1}
\end{gather}
Acting by the f\/irst row of $J_1$ on the column of variational derivatives of~$H_0$ in~\eqref{be_Ham} and app\-lying~$\nabla_c$ we obtain
\begin{gather}
 \nabla_c\big(J_1^{11}\delta_u H_0 + J_1^{12} \delta_v H_0\big) \equiv c_3D_2\delta_vH_0 - c_1\delta_uH_0 = \nabla_c(v) = c_1v_1 + c_2v_2. \label{1stline}
\end{gather}
On account of the relation \eqref{J1B1}, the second row of the last equation in \eqref{be_Ham} reads
\begin{gather*}
 J_1^{21}\delta_u H_0 + J_1^{22}\delta_v H_0 \equiv \frac{1}{c_1}B_1\delta_v H_0 + \frac{c_3}{c_1}D_2\nabla_c^{-1}(c_3D_2\delta_v H_0 - c_1\delta_u H_0) = v_t, 
\end{gather*}
which with the use of \eqref{1stline} becomes
\begin{gather}
 \frac{1}{c_1}B_1(\delta_v H_0) + \frac{c_3}{c_1} v_2 = v_t\quad\iff\quad B_1(\delta_v H_0) = c_1v_t - c_3v_2. \label{line2}
\end{gather}
We assume a linear dependence of $H_0$ on $v$
\begin{gather}
 H_0 = b[u]v + c[u]\quad\Longrightarrow\quad \delta_v H_0 = \frac{\partial H_0}{\partial v} = b[u], \label{H_v}
\end{gather}
where $b$ and $c$ depend only on $u$ and its derivatives.

We note that adding to $H_0$ the term $av^2$ with constant $a$ does not contribute to the Hamiltonian f\/low
\begin{gather*}J_1
\left(
\begin{matrix}
 \delta_u H_0\\
 \delta_v H_0
\end{matrix}
\right)\end{gather*}
and hence this is unnecessary.

Plugging \eqref{H_v} into equation \eqref{line2} and using the def\/inition of $B_1$ from \eqref{AiBi}, we obtain
\begin{gather}
 c_1\big\{c_3(u_{22}D_1[b] - u_{12}D_2[b]) + c_1(v_1D_2[b] - v_2D_1[b])\big\} + c_1c_6D_1[b]
\nonumber\\
 \qquad{} + (c_1c_7 - c_2c_6 + c_3c_4)D_2[b] = c_1\big\{v_1\nabla_c(u_2) - v_2\nabla_c(u_1) + c_3\big(u_{11}u_{22}-u_{12}^2\big) \nonumber\\
\qquad {} + c_4v_1 + c_5v_2 + c_6u_{11} + c_7u_{12} + c_8u_{22} + c_9\big\} - c_3v_2. \label{line2terms}
\end{gather}
Splitting equation \eqref{line2terms} with respect to $v_1$ and $v_2$ and collecting separately terms with~$v_1$ and~$v_2$ implies the following two equations
\begin{gather}
 D_1[b] = \frac{1}{c_1} \left\{ \nabla_c(u_1)+ \frac{(c_3 - c_1c_5)}{c_1} \right\}, \qquad D_2[b] = \frac{1}{c_1} (\nabla_c(u_2) + c_4 ). \label{D1b D2_b}
\end{gather}
Integrating these two equations we obtain
\begin{gather}
b = \frac{1}{c_1} \nabla_c(u) + \frac{c_4}{c_1}z_2 + \frac{(c_3 - c_1c_5)}{c_1^2}z_1 + s_0, \label{b}
\end{gather}
where $s_0$ is a constant of integration. Plugging~\eqref{D1b D2_b} into the remaining terms in \eqref{line2terms}, we note that the terms $c_1c_3(u_{11}u_{22}-u_{12}^2)$ and all terms proportional to $u_{11}$ and $u_{12}$ are canceled identically while the terms proportional to~$u_{22}$ cancel due to integrability condition~\eqref{intcon}. The remaining constant terms in~\eqref{line2terms} imply the relation~\eqref{c9} which is an additional (``Hamiltonian'') condition for integrable system \eqref{2comp} to have bi-Hamil\-tonian form.

Using the result \eqref{b} for $b$ in our ansatz \eqref{H_v} for $H_0$ yields
\begin{gather}
 H_0 = v\left\{\frac{1}{c_1} \nabla_c(u) + \frac{c_4}{c_1}z_2 + \frac{(c_3 - c_1c_5)}{c_1^2}z_1 + s_0\right\} + c[u]. \label{H0}
\end{gather}
From \eqref{1stline} we have
\begin{gather}
\delta_u H_0 \equiv \delta_u(b[u] v) + \delta_u(c[u]) = \frac{1}{c_1}\big\{c_3D_2[b] - \nabla_c(v)\big\}, \label{del_uH0}
\end{gather}
where \eqref{b} implies that
\begin{gather*}\delta_u(b v) = -\frac{1}{c_1}\nabla_c(v).\end{gather*}
After cancelations, equation \eqref{del_uH0} yields
\begin{gather}
 \delta_u(c[u]) = \frac{c_3}{c_1}D_2[b] \equiv \frac{c_3}{c_1^2}\big(\nabla_c(u_2) + c_4\big), \label{del_uc}
\end{gather}
where $D_2[b]$ from \eqref{D1b D2_b} is used. This result obviously suggests $c[u]$ to be of the form
\begin{gather}
 c[u] = \alpha u_1u_2 + \beta u_2^2 + \gamma(u) \label{c[u]form}
\end{gather}
with constant $\alpha$ and $\beta$, so that $\delta_u(c) = - D_1(\alpha u_2) - D_2(\alpha u_1 + 2\beta u_2) + \gamma'(u)$.
Plugging the latter expression in the equation \eqref{del_uc} we f\/ind the coef\/f\/icients in~\eqref{c[u]form} to be
\begin{gather*}\alpha = -\frac{c_3}{2c_1},\qquad \beta = -\frac{c_2c_3}{2c_1^2},\qquad \gamma(u) = \frac{c_3c_4}{c_1^2}u + \gamma_0\end{gather*}
with the constant of integration $\gamma_0$, so that $c[u]$ in \eqref{c[u]form} becomes
\begin{gather*}
 c[u] = - \frac{c_3}{2c_1^2}u_2\nabla_c(u) + \frac{c_3c_4}{c_1^2}u + \gamma_0. 
\end{gather*}
We have now completely determined the Hamiltonian density $H_0$ in \eqref{H0} to be \eqref{H0full}, where we have skipped the additive constant~$\gamma_0$.
\end{proof}

With the Hamiltonian density $H_0$ from \eqref{H0full}, corresponding to the second Hamiltonian ope\-ra\-tor~$J_1$, the system~\eqref{2comp} admits bi-Hamiltonian representation~\eqref{be_Ham}, provided that the integrability condition~\eqref{intcon} and bi-Hamiltonian constraint~\eqref{c9} are satisf\/ied in the generic case. Thus, we obtain the f\/irst seven-parameter family of bi-Hamiltonian systems.

\subsubsection{Second family of bi-Hamiltonian systems}

Composing an alternative recursion operator $R'$, determined by~\eqref{R'} and \eqref{R'ij}, with the f\/irst Hamiltonian operator~$J_0$ we obtain an alternative second Hamiltonian operator $J'_1=R'J_0$ with the resulting expression
\begin{gather}
 J_1' = \left(
 \begin{matrix}
 - c_2\nabla_c^{-1} & - c_3\nabla_c^{-1}D_1 \vspace{1mm}\\
 c_3\nabla_c^{-1}D_1 & \displaystyle\frac{1}{c_2}\left(c_3^2\nabla_c^{-1}D_1^2 + B_1)\right)
 \end{matrix} \right), \label{J'}
\end{gather}
where $B_1$ is def\/ined in \eqref{AiBi2}. Next task is to determine the Hamiltonian density $H_0'$ corresponding to~$J_1'$ in the bi-Hamiltonian representation
\begin{gather}
 J_0
\left(
\begin{matrix}
 \delta_u H_1\\
 \delta_v H_1
\end{matrix}
\right) = J'_1
\left(
\begin{matrix}
 \delta_u H'_0\\
 \delta_v H'_0
\end{matrix}
\right) = \left(
\begin{matrix}
 v\\
 v_t
\end{matrix}
\right). \label{be_Ham'}
 \end{gather}
A check of the Jacobi identities for $J'_1$ and compatibility of the two Hamiltonian structures~$J_0$ and~$J'_1$ is straightforward and too lengthy to be presented here.
\begin{Proposition}
Bi-Hamiltonian representation \eqref{be_Ham'} of the system \eqref{2comp} is valid under the constraint
\begin{gather}
 c_9 = \frac{1}{c_2^2}\big\{c_8(c_2c_4 + c_3) + c_5(c_1c_8 - c_2c_7 + c_3c_5)\big\} \label{c9'}
\end{gather}
 with the following Hamiltonian density
\begin{gather}
 H'_0 = \frac{v}{c_2}\left\{\nabla_c(u) - c_5 z_1 + \left(c_4 + \frac{c_3}{c_2}\right)z_2\right\} + s_0v + \frac{c_3}{2c_2^2}\big(u_1\nabla_c(u) + 2c_5u\big). \label{H'_02}
\end{gather}
\end{Proposition}
\begin{proof}
Acting by the f\/irst row of $J'_1$ on the column of variational derivatives of $H'_0$ in \eqref{be_Ham'} and applying $\nabla_c$ we obtain
\begin{gather}
 \delta_uH'_0 = -\frac{1}{c_2}\big(c_3D_1\delta_vH'_0 + \nabla_c v\big). \label{1stline'}
\end{gather}
The second row of equation \eqref{be_Ham'} becomes
\begin{gather*}
c_3\nabla_c^{-1}D_1\delta_u H'_0 + \frac{1}{c_2}\big(c_3^2\nabla_c^{-1}D_1^2 + B_1\big)\delta_v H'_0 = v_t, 
\end{gather*}
which with the use of \eqref{1stline'} becomes
\begin{gather}
 B_1(\delta_v H'_0) = c_2v_t + c_3v_1.\label{line2'}
\end{gather}
We again assume $H'_0$ to be linear in $v$
\begin{gather}
 H'_0 = b[u]v + c[u]\quad\Longrightarrow\quad \delta_v H'_0 = \frac{\partial H'_0}{\partial v} = b[u], \label{H'_v}
\end{gather}
where $b$ and $c$ depend only on $u$ and its derivatives. Plugging~\eqref{H'_v} into equation \eqref{line2'}, using the expression for $B_1$ from~\eqref{AiBi2} and $v_t$ from our equation \eqref{2comp}, we obtain the equation linear in $v_1$ and $v_2$, similar to~\eqref{line2terms}. Equating the coef\/f\/icients of~$v_1$ and~$v_2$ in this linear equation on both sides, we obtain
\begin{gather}
 D_1[b] = \frac{1}{c_2} \{\nabla_c(u_1) - c_5\},\qquad D_2[b] = \frac{1}{c_2} \left\{\nabla_c(u_2) + c_4 + \frac{c_3}{c_2}\right\}. \label{D1b D2_b'}
\end{gather}
Integrating these equations we obtain
\begin{gather}
 b = \frac{1}{c_2}\left\{\nabla_c(u) - c_5 z_1 + \left(c_4 + \frac{c_3}{c_2}\right)z_2\right\} + s_0. \label{b'}
\end{gather}
Plugging the expressions \eqref{D1b D2_b'} in remaining terms of the equation linear in $v_1$ and $v_2$, we f\/ind that all variable terms cancel and we end up with the relation \eqref{c9'} among constant terms. This is again an additional condition for integrable system \eqref{2comp} to have bi-Hamiltonian form which is in general dif\/ferent from \eqref{c9} obtained for our f\/irst choice of the operators~$A_i$,~$B_i$. Thus, we obtain two dif\/ferent families of seven-parameter bi-Hamiltonian systems.

Plugging the expression \eqref{b'} for $b$ into the formula \eqref{H'_v} for $H'_0$, calculating $\delta_uH'_0$ and using equation \eqref{1stline'}, we obtain
\begin{gather*}
 \delta_uc[u] = - \frac{c_3}{c_2^2}(c_1u_{11} + c_2u_{12} - c_5), 
\end{gather*}
which implies
\begin{gather}
 c[u] = \frac{c_3}{2c_2^2}(u_1\nabla_c(u) + 2c_5u) + \gamma_0.
 \label{c'(u)}
\end{gather}
Using our results \eqref{b'} for $b$ and \eqref{c'(u)} for $c[u]$ in the formula \eqref{H'_v}, we obtain the Hamiltonian density $H'_0$ to be \eqref{H'_02} for the alternative bi-Hamiltonian system~\eqref{be_Ham'} in the generic case. Here again we have skipped the additive constant~$\gamma_0$.
\end{proof}

\subsubsection{A family of tri-Hamiltonian systems}

If we impose an additional condition on the coef\/f\/icients such that both constraints~\eqref{c9} and \eqref{c9'} coincide, we eliminate $c_9$ and consider the resulting condition together with the integrability condition~\eqref{intcon} and the constraint~\eqref{c9}. With these three conditions, we obtain a six-parameter family of tri-Hamiltonian systems
\begin{gather*}
\left(
\begin{matrix}
 u_t\\
 v_t
\end{matrix}
\right) = J_0
\left(
\begin{matrix}
 \delta_u H_1\\
 \delta_v H_1
\end{matrix}
\right) = J_1
\left(
\begin{matrix}
 \delta_u H_0\\
 \delta_v H_0
\end{matrix}
\right) = J'_1
\left(
\begin{matrix}
 \delta_u H'_0\\
 \delta_v H'_0
\end{matrix}
\right) 
 \end{gather*}
with compatibility of the three Hamiltonian operators $J_0$, $J_1$, $J'_1$ been checked. An explicit solution of the two constraints together with the integrability condition \eqref{intcon} for the tri-Hamiltonian system has the form
\begin{gather*}
 c_4 = \frac{1}{c_2c_3}\big(c_2^2c_6 - c_1^2c_8 - c_1c_3c_5\big),\qquad c_7 = \frac{1}{c_1c_2} \big(2c_1^2c_8 + 2c_1c_3c_5 - c_3^2\big),\nonumber\\
c_9 = \frac{(c_3c_5 + c_1c_8)}{c_1c_2^2c_3}\big\{c_3(c_3 - c_1c_5) - c_1^2c_8\big\} + \frac{c_6c_8}{c_3} 
\end{gather*}
with the six free parameters $c_1$, $c_2$, $c_3$, $c_5$, $c_6$, $c_8$.

\subsection{Particular cases}

The f\/irst case $c_1=0$, $c_2\ne 0$ we should treat as a particular case of the second bi-Hamiltonian family. Therefore, we put $c_1=0$ in the second Hamiltonian operator $J'_1$ from~\eqref{J'}
\begin{gather*}
 J'_1 = \left(
 \begin{matrix}
 - D_2^{-1} & - \displaystyle \frac{c_3}{c_2}D_2^{-1}D_1 \vspace{1mm}\\
 \displaystyle \frac{c_3}{c_2}D_2^{-1}D_1 & \displaystyle \frac{c_3^2}{c_2^2}D_2^{-1}D_1^2 + \frac{1}{c_2}B_1
 \end{matrix}
 \right), 
\end{gather*}
where $B_1$ is def\/ined in \eqref{AiBic1=0}, so that explicitly we have
\begin{gather*}
 (J'_1)^{22} = - (c_3L_{12(1)} + c_2L_{12(t)}) + \left(c_7 - \frac{c_3c_5}{c_2}\right) D_1 + c_8 D_2.
\end{gather*}
The corresponding Hamiltonian density is a particular case of the density $H'_0$ from \eqref{H'_02} at $c_1=0$
\begin{gather*}
 H'_0 = v\left\{u_2 - \frac{c_5}{c_2} z_1 + \frac{c_6}{c_3}z_2\right\} + s_0v + \frac{c_3}{2c_2}u_1u_2 + \frac{c_3c_5}{c_2^2} u,
\end{gather*}
where the integrability condition $c_3(c_2c_4+c_3)=c_2^2c_6$ has been used and an additional constraint has the form
\begin{gather*}
 c_9 = \frac{1}{c_2^2}\{c_8(c_2c_4 + c_3) + c_5(c_3c_5 - c_2c_7)\}. 
\end{gather*}

The second case $c_2=0$, $c_1\ne 0$ is treated as a particular case of the f\/irst bi-Hamiltonian family, so we put $c_2=0$ in the second Hamiltonian operator $J_1$ from \eqref{J_1}
\begin{gather*}
 J_1 =\left(
 \begin{matrix}
 - D_1^{-1} & \displaystyle \frac{c_3}{c_1}D_1^{-1}D_2 \\
 - \displaystyle \frac{c_3}{c_1}D_1^{-1}D_2 & \displaystyle \frac{c_3^2}{c_1^2}D_1^{-1}D_2^2 + \frac{1}{c_1}B_1
 \end{matrix}
 \right), 
\end{gather*}
where $B_1$ is def\/ined in \eqref{AiBic2=0}, so that we have explicitly
\begin{gather*}
 (J_1)^{22} = \frac{c_3^2}{c_1^2}D_1^{-1}D_2^2 + c_3L_{12(2)} - c_1L_{12(t)}) + c_6D_1 + \left(c_7 + \frac{c_3c_4}{c_1}\right) D_2.
\end{gather*}
The corresponding Hamiltonian density is a particular case of the density $H_0$ from \eqref{H0full} at $c_2=0$
\begin{gather*}
H_0 = v\left(u_1 + \frac{c_4}{c_1}z_2 + \frac{c_8}{c_3}z_1 + s_0\right) - \frac{c_3}{2c_1}u_1u_2 + \frac{c_3c_4}{c_1^2} u,
\end{gather*}
where integrability condition \eqref{spec2'} has been used, with the additional constraint
\begin{gather*}
 c_9 = \frac{1}{c_1^2}\big\{c_6(c_3 - c_1c_5) + c_4(c_1c_7 + c_3c_4)\big\}. 
\end{gather*}

\subsection[Special case $c_3=0$]{Special case $\boldsymbol{c_3=0}$}

The second Hamiltonian operator is a composition $J_1 = RJ_0$ of the recursion operator \eqref{R1} for this special case with the f\/irst Hamiltonian operator~$J_0$. Multiplication of these two matrices yields
\begin{gather*}
 J_1 = \left(
 \begin{matrix}
 c_2\nabla_c^{-1} & 0 \\
 0 & \displaystyle\frac{1}{c_1} B_1
 \end{matrix}
 \right),
\end{gather*}
where
\begin{gather*}
 J_1^{22} = - c_2\left(v_1D_2 - v_2D_1 + \frac{c_6}{c_1}D_1 + \frac{c_8}{c_2}D_2\right) = \frac{1}{c_1} B_1 
\end{gather*}
with $B_1$ def\/ined in \eqref{AiBispec1}.

The remaining problem is to obtain the Hamiltonian density $H_0$ corresponding to the second Hamiltonian operator $J_1$ according to~\eqref{2comp}
\begin{gather}
 J_1
\left(
\begin{matrix}
 \delta_u H_0\\
 \delta_v H_0
\end{matrix}
\right) = \left(
\begin{matrix}
 v\\
 v_t
\end{matrix}
\right). \label{be_Ham1}
 \end{gather}
\begin{Proposition}
Bi-Hamiltonian representation \eqref{be_Ham1} of the system \eqref{2comp} is valid under the constraint
\begin{gather}
 c_9 = \frac{(c_1c_4c_8-c_2c_5c_6)}{c_1c_2} \label{c9_1}
\end{gather}
 with the following Hamiltonian density
\begin{gather}
 H_0 = \frac{v}{c_2}\{-\nabla_c(u) + c_5z_1 - c_4z_2 + s_0\}. \label{H01}
\end{gather}
\end{Proposition}
\begin{proof}
The equation following from the f\/irst row of the matrix equation \eqref{be_Ham1} is $c_2\nabla_c^{-1}\delta_u H_0 = v$ or equivalently
\begin{gather}
 \delta_u H_0 = \frac{1}{c_2}\nabla_c(v) = \frac{c_1}{c_2}v_1 + v_2.\label{1st}
\end{gather}
The second row of \eqref{be_Ham1} yields
\begin{gather}
\frac{1}{c_1} B_1(\delta_vH_0) = v_t. \label{2nd}
\end{gather}
We assume again a linear dependence of $H_0$ on $v$
\begin{gather}
 H_0 = b[u]v + c[u] \quad \Longrightarrow \quad \delta_vH_0 = b[u], \label{H_0v}
\end{gather}
where $b$ and $c$ may depend on $u$ and its partial derivatives.

 Using \eqref{H_0v} and also using the def\/inition of $B_1$ and $v_t$, given by \eqref{2comp}, on the left-hand side and right-hand side of \eqref{2nd}, respectively, we obtain
\begin{gather}
 c_2(v_1D_2[b] - v_2D_1[b]) + \frac{c_2c_6}{c_1}D_1[b] + c_8D_2[b]\nonumber\\
 \qquad{} = - \big\{v_1\nabla_c(u_2) - v_2\nabla_c(u_1) + c_4v_1 + c_5v_2 + c_6u_{11} + c_7u_{12} + c_8u_{22} + c_9\big\}. \label{2nda}
\end{gather}
Splitting this equation with respect to $v_1$ and $v_2$ we obtain two equations
\begin{gather*} D_1[b] = -\frac{1}{c_2}\nabla_c(u_1) + \frac{c_5}{c_2},\qquad D_2[b] = -\frac{1}{c_2}\nabla_c(u_2) - \frac{c_4}{c_2}, \end{gather*}
which are integrated in the form
\begin{gather}
 b = \frac{1}{c_2}\big\{{-}\nabla_c(u) + c_5z_1 - c_4z_2 + s_0\big\}. \label{b1}
\end{gather}
Using \eqref{b1}, we check that all the remaining nonconstant terms in \eqref{2nda} cancel and we obtain an additional relation~\eqref{c9_1} between coef\/f\/icients. Plugging the expression~\eqref{b1} for~$b$ into~\eqref{H_0v} for~$H_0$ we obtain
\begin{gather}
 H_0 = \frac{v}{c_2}\big\{{-}\nabla_c(u) + c_5z_1 - c_4z_2 + s_0\big\} + c[u]. \label{H0_1}
\end{gather}
Computing $\delta_u H_0$ from \eqref{H0_1} and plugging the result into \eqref{1st}, we conclude that $\delta_u c[u] = 0$ and hence we may choose $c[u] = 0$. Thus, we end up with the formula \eqref{H01}.
\end{proof}

\section{Conclusion}

We have shown that all the Euler--Lagrange equations of the evolutionary Hirota type in $(2+1)$ dimensions have the symplectic Monge--Amp\`ere form. The symmetry condition for such an equation is converted to a skew-factorized form. Then recursion relations and Lax pairs are obtained as immediate consequences of this representation. We have converted the equation into a two-component evolutionary form and obtained Lagrangian and recursion operator for this two-component system. The Lagrangian is degenerate because the momenta cannot be inverted for the velocities.
 Applying to this degenerate Lagrangian the Dirac's theory of constraints, we have obtained a symplectic operator and its inverse, the latter being a Hamiltonian opera\-tor~$J_0$. We have found the corresponding Hamiltonian density~$H_1$, thus presenting our system in a~Hamiltonian form.
Composing the recursion operator~$R$ with~$J_0$, we have obtained the second Hamiltonian operator $J_1=RJ_0$. We have found the Hamiltonian density $H_0$ corresponding to $J_1$ and thereby obtained bi-Hamiltonian representation of our system under one constraint on the coef\/f\/icients, which is additional to the integrability condition. Thus, we end up with a~seven-parameter class of bi-Hamiltonian systems in $(2+1)$ dimensions.

Representation of symmetry condition in an alternative skew-factorized form produces another seven-parameter family of bi-Hamiltonian systems which satisfy another constraint, additional to the integrability condition. This constraint is dif\/ferent from the additional constraint for the f\/irst family of bi-Hamiltonian systems. If we require both constraints to coincide, then, accounting also for the integrability condition, we end up with a six-parameter family of tri-Hamiltonian systems.

We apply similar methods to Hirota equations in $3+1$ dimensions with similar results. This work is in progress and will be published elsewhere.

\subsection*{Acknowledgments}

The authors are grateful to an anonymous referee for important remarks which contributed to the improvement of our paper. The research of M.B.~Sheftel is partly supported by the research grant from Bo\u{g}azi\c{c}i University Scientif\/ic Research Fund (BAP), research project No.~11643.

\pdfbookmark[1]{References}{ref}
\LastPageEnding

\end{document}